\documentclass[runningheads]{llncs}
\usepackage[pdftex]{graphicx,color}
\usepackage[noend]{algorithmic}
\usepackage{algorithm}
\usepackage[czech,english]{babel}
\usepackage[latin2]{inputenc}
\usepackage{verbatim}
\usepackage{wrapfig}
\usepackage{paralist}

\usepackage{amssymb,amsmath,amsfonts}
\usepackage{enumerate}
\usepackage{wrapfig}
\usepackage{xspace}

\usepackage{tikz,pgffor}
\usetikzlibrary{arrows}
\usetikzlibrary{shapes}
\usetikzlibrary{automata}
\usetikzlibrary{decorations.pathmorphing} % for snake lines
\usetikzlibrary{decorations.pathreplacing} % for snake lines
\usetikzlibrary{calc} % for manimulation of coordinates

\newcommand{\events}{E}
\newcommand{\order}{<}
\newcommand{\procSet}{{\cal P}}		
\newcommand{\type}{{\cal T}}
\newcommand{\proc}{{P}}
\newcommand{\match}{{\cal M}}
\newcommand{\stateSet}{{\cal S}}
\newcommand{\rel}{\tau}
\newcommand{\MSCmap}{{\it L}}
\newcommand{\automaton}{{\cal A}}
\newcommand{\channel}{{\cal B}}
\newcommand{\action}{{A}}
\newcommand{\transition}{\rightarrow}
\newcommand{\final}{F}

\newcommand{\dnlc}{{\cal U}}
\newcommand{\local}{{\cal L}}

\newcommand{\labels}{{l}}
\newcommand{\content}{{\cal C}}

\newcommand{\lang}{{\cal L}}

\begin{document}

\pagestyle{headings}

\mainmatter              % start of the contributions

%\title{A Bigger Class of Realizable Message Sequence Graphs}
\title{Controllable-choice Message Sequence Graphs}
%\title{Negotiable-choice Message Sequence Graphs and realizability}
%Realizability of Message~Sequence~Graphs}

%\titlerunning{A Bigger Class of Realizable MSG}
\titlerunning{Controllable-choice MSG}
%\titlerunning{Negotiable-choice MSG}
%Hamiltonian Mechanics}  % abbreviated title (for running head)

\author{Martin Chmel\'{i}k\inst{1}\thanks{The author was supported by Austrian Science Fund (FWF) Grant No P 23499-N23 on Modern Graph Algorithmic Techniques in Formal Verification, FWF NFN Grant No S11407-N23 (RiSE), ERC Start grant (279307: Graph Games), and Microsoft faculty fellows award.} and Vojt\v{e}ch
  {\v{R}}eh\'{a}k\inst{2}\thanks{The author is supported by the Czech
    Science Foundation, grant No.~P202/12/G061.}} 
%\author{Martin Chmel\'{i}k}

\authorrunning{M. Chmel\'{i}k and V. {\v{R}}eh\'{a}k}   % abbreviated author list (for running head)

\institute{ 
Institute of Science and Technology Austria (IST Austria), Klosterneuburg,
Austria\\\email{martin.chmelik@ist.ac.at}
\and
Faculty of Informatics Masaryk University, Brno, Czech
  Republic\\\email{rehak@fi.muni.cz}
}

% typeset the title of the contribution
\maketitle              

\begin{abstract}
  We focus on the realizability problem of Message Sequence Graphs (MSG),
  i.e.\ the problem whether a given MSG specification is correctly
  distributable among parallel components communicating via messages.  This
  fundamental problem of MSG is known to be undecidable. We introduce a well
  motivated restricted class of MSG, so called controllable-choice MSG, and
  show that all its models are realizable and moreover it is decidable
  whether a given MSG model is a member of this class. In more detail, this
  class of MSG specifications admits a deadlock-free realization by
  overloading existing messages with additional bounded control data.  We
  also show that the presented class is the largest known subclass of MSG
  that allows for deadlock-free realization.
\end{abstract}

\section{Introduction}

Message Sequence Chart (MSC) \cite{ITU-MSC'04} is a popular formalism for
specification of distributed systems behaviors (e.g. communication
protocols or multi-process systems). Its simplicity and intuitiveness
come from the fact that an MSC describes only exchange of messages between system
components, while other aspects of the system (e.g. content of the messages
and internal computation steps) are abstracted away.  The formalism consists
of two types of charts: (1)~basic Message Sequence Charts (bMSC) that are
suitable for designing finite communication patterns and (2)~High-level
Message Sequence Charts (HMSC) combining bMSC patterns into more complex
designs. In this paper, we focus on the further type reformulated as Message
Sequence Graphs (MSG) that has the same expressive power as HMSC but a simpler
structure, and hence it is often used in theoretical computer science
papers, see, e.g.\ \cite{MSC09,MSC01,MSC04,BOL08,gt_vmt_postproc}.

Even such incomplete models as MSG can indicate serious errors in the
designed system. The errors can cause problems during implementation or even
make it impossible. Concerning verification of MSC models, researchers have
studied a presence of a~race condition in an MSC
\cite{MSC02,MSC13,MSC27,gt_vmt_postproc}, boundedness of the~message
channels \cite{MSC01}, the~possibility to reach a~non-local branching node
\cite{MSC04,MSC35,BOL08,jard:inria-00584530,BOL01,BOL17,MSC45}, deadlocks,
livelocks, and many more. For a recent
overview of current results see, e.g. \cite{Dan2012}.

In this paper, we focus on the realizability problem of MSG
specifications,
i.e. implementation of the specification among parallel machines
communicating via messages. This problem has been studied in various
settings reflecting parameters of the parallel machines, the environment
providing message exchanges as well as the type of equivalence considered
between the MSG specification and its implementation. Some authors restricted
the communication to synchronous handshake \cite{BOL08,BOL19}, needed
several initial states in the synthesized machines \cite{MSC49}, or
considered language equivalence with global accepting states in the
implementation (the implementation accepts if the components are in specific
combinations of its states) \cite{MSC46}. From our point of view, the
crucial aspect is the attitude to non-accepted executions of the
implementation.  When language equivalence is taken into account, an
intentional deadlock can prevent a badly evolving execution from being
accepted \cite{BOL08}. In our setting every partial execution can be
extended into an accepting one. Therefore, we focus on a
deadlock-free implementation of a given MSG into Communicating Finite-State
Machines (CFM) with FIFO communicating channels and distributed acceptance
condition, i.e. a CFM accepts if each machine is in an accepting state. In
\cite{MSC41}, it has been shown that existence of a CFM realizing a given
MSG without deadlocks is undecidable. When restricted to \emph{bounded} MSG
(aka regular MSG, i.e. communicating via finite/bounded channels, and so
generating a regular language), the problem is EXPSPACE-complete
\cite{MSC41}.

In later work \cite{BOL08,MSC49}, a finite data extension of
messages was considered when realizing MSG. This is a very natural
concept because message labels in MSG are understood as message types
abstracting away from the full message content. Hence, during
implementation, the message content can be refined with additional (finite)
data that helps to control the computation of the CFM in order to achieve the
communication sequences as specified in the given MSG. The main obstacle when
realizing MSG are nodes with multiple outgoing edges --- choice nodes. In
a~CFM realization, it is necessary to ensure that all Finite-State Machines
 choose the~same successor of each choice node.
This can be hard to achieve as the~system is distributed.

In \cite{BOL08}, a class of so called \emph{local-choice} MSG
\cite{BOL01,MSC04} was shown to include only MSG realizable in the above
mentioned setting.
% In \cite{BOL08}, the class of so called \emph{local-choice MSG}
% \cite{BOL01,MSC04} was shown to be realizable in the above mentioned setting.
Local-choice specifications have the~communication after each choice node
initiated by a~single process --- the~choice leader.  Intuitively, whenever
a~local-choice node is reached, the~choice leader machine attaches to all its
outgoing messages the~information about the~chosen node. The other machines pass the~information on. This construction is
sufficient to obtain a~deadlock-free realization, for details see
\cite{BOL08}. Another possible realization of local-choice MSG is
presented in \cite{jard:inria-00584530}.
 Due to \cite{BOL17}, it is also decidable to determine whether a~given MSG is
language equivalent to some local-choice MSG and, moreover, each equivalent
MSG can be algorithmically realized by a CFM.
To the~best of our knowledge, this is the~largest class of deadlock-free realizable
specifications in the~standard setting, i.e. with additional data, FIFO
channels, and local accepting states.

% Nodes that are not local-choice were identified as a potential MSG
% design problem first in \cite{BOL01} and are further studied in \cite{MSC04}.

In this paper, we introduce a  new class of \emph{controllable-choice} MSG that
extends this large class of realizable MSG. The crucial idea of
controllable-choice MSG is that even some non-local-choice nodes can be
implemented, if the processes initiating the communication after the
choice can agree on the successor node in advance. This is achieved by exchanging bounded additional content
in existing messages. We call choice nodes where such an agreement is
possible \emph{controllable-choice} nodes, and show that the class of MSG with
these nodes is more expressive than the class of MSG that are
language equivalent to local-choice MSG. 

\section{Preliminaries}
\label{chap:foundations}

In this section, we introduce the~Message Sequence Chart (MSC) formalism that
was standardized by the~International Telecommunications Union (ITU-T) as
Recommendation Z.120 \cite{ITU-MSC'04}. It is used to model interactions
among parallel components in a~distributed environment. First, we introduce
the basic MSC.

\subsubsection{basic Message Sequence Charts (bMSC)}
Intuitively, a bMSC identifies a~single finite execution of
a~message passing system. Processes are denoted as vertical lines ---
instances. Message exchange is represented by an arrow from the~sending
process to the~receiving process.  Every process identifies a~sequence of
actions --- sends and receives --- that are to be executed in the~order from
the~top of the~diagram. The communication among the~instances is not
synchronous and can take arbitrarily long time.
\begin{definition}
A \emph{basic Message Sequence Chart (bMSC)} $M$ is defined by a~tuple
  $(\events,\order,\procSet,\type,\proc,\match, \labels)$ where:
  \begin{compactitem}
  \item $\events$ is a~finite set of events,
  \item $\order$ is a~partial ordering on $\events$ called visual order,
  \item $\procSet$ is a~finite set of processes,
  \item $\type: \events \rightarrow \{$send$, $receive$\}$ is a~function
    dividing events into sends and receives,
  \item $\proc: \events \rightarrow \procSet$ is a~mapping that associates
    each event with a~process,
  \item $\match: \type^{-1}(\mathrm{send}) \rightarrow
    \type^{-1}(\mathrm{receive})$ is a bijective mapping, relating every
    send with a~unique receive, such that % for any $e$ and $f= \match(e)$, we have
    % $\proc(e)\neq\proc(f)$,i.e.,
    a~process cannot send a message to itself, we refer to a pair of events $(e,
\match(e))$ as a \emph{message}, and
  \item $\labels$ is a function associating with every message $(e,f)$ a label
    $m$ from a finite set of message labels $\content$, i.e. $l(e,f)=m$.
%$\content$ is a~finite set of message labels, and $\labels$ is a function
%    associating with every message $(e,\match(e))$ its label
%    $l(e,\match(e))$ from $\content$.
  \end{compactitem}
  Visual order $\order$ is defined as the~reflexive and transitive closure
  of $ \mathcal{M} \cup \bigcup_{p \in \mathcal{P}} \order_{p}$ where
  $\order_{p}$ is a~total order on $\proc^{-1}(p)$.
\end{definition}
We require the bMSC to be \emph{first-in-first-out (FIFO)}, i.e., the~visual order
satisfies for all messages $(e,f),(e', f')$ and processes $p,p'$ the~following condition
$$e <_{p} e' \wedge  \proc(f) = \proc(f') = p' \; \Rightarrow \; f <_{p'} f'.$$
%For the ease of presentation, we will consider only FIFO bMSCs in this
%work. The results can be naturally extended to non-FIFO bMSCs.
%\todo{proc? bude to fungovat i pro ostatni? jsou tedy duvody jen prezentacni, technicke, nebo fatalni(tj. pro jine dukaz nemame a nevime, jestli to plati)?}
% Examples of bMSCs can be seen in Figure~\ref{fig:concatenate_msc}.
%\todo{nejdriv jazyk $\Sigma$, potom linearizace}
Every event of a bMSC can be represented by a letter from an alphabet % $\Sigma$:
$$\Sigma = \{p!q(m) \mid p,q \in \procSet, m \in \content\} \cup
\{q?p(m) \mid p,q \in \procSet ,m \in \content\}.$$ Intuitively, $p!q(m)$
denotes a~send event of a message with a label $m$ from a~process~$p$ to
a~process~$q$, and $q?p(m)$ represents a~receive event of a message with
a~label $m$ by $q$ from a~process~$p$.  We define a~\emph{linearization} as
a word over $\Sigma$ representing a total order of events that is consistent
with the partial order $\order$.
% A~total ordering of events from $\Sigma$ that is extending the partial
% ordering $\order$ is called a~\emph{linearization}.
For a given bMSC $M$, a \emph{language}  $\lang(M)$ is
the set of all linearizations of $M$.
%Once the~partial ordering of the~events is defined, we can define
%a~linearization --- a~total order that is an extension of the~partial order
%$\order$. We will represent linearization as a~string over alphabet
%$\Sigma$.
%$$ \Sigma = \{p!q(m) \mid p,q \in \procSet, m \in \content\} \cup \{q?p(m)
%\mid p,q \in \procSet ,m \in \content\} $$

%Where $p!q(m)$ denotes the~send event of message $m$ from process $p$ to
%process $q$ and $q?p(m)$ represents a~receive event of message $m$ by $q$
%from process $p$. The set of all linearizations for a~given bMSC $M$
%is denoted as $\lang(M)$.

%\begin{example}
%Let us consider the~bMSC $M_{3}$ in Figure~\ref{fig:concatenate_msc}, then the~linearization set for bMSC $M_{3}$ is equal to:

%\begin{eqnarray*}
%\lang(M_{3})=\{ & p!q(a) \: q?p(a) \: p!q(b) \: q?p(b), &\\
%& p!q(a) \: p!q(b) \: q?p(a) \: q?p(b) & \}
%\end{eqnarray*}

%%$$ \lang(M_{3}) = \{ p!q(a) \: q?p(a) \: p!q(b) \: q?p(b),  p!q(a) \: p!q(b) \: q?p(a) \: q?p(b)\}$$
%\end{example}

%In the following section, we define a
%high-level specification concept that allows for more complex constructions
%such as branching to altertnative secification paterns as well as iteration
%and calls of referenced patterns.

\subsubsection{Message Sequence Graphs}
It turns out that specifying finite communication patterns is not sufficient
for modelling complex systems. 
Message Sequence Graphs allow us to combine bMSCs into more complex
systems using alternation and iteration. An MSG is a~directed graph
with nodes labeled by bMSCs and two special nodes, the~initial
and the~terminal node.  Applying the concept of finite automata
\cite{Gill-book_62}, the graph represents a set of paths from the initial
node to the terminal node. In MSG, every such a path identifies a~sequence
of bMSCs. As every finite sequence of bMSCs can be composed
into a~single bMSC, an MSG serves as a~finite representation
of an (infinite) set of bMSCs.

\begin{definition}
A \emph{Message Sequence Graph (MSG)} is defined by a~tuple $ G=(\stateSet,\rel,s_{0},
  s_{f},\MSCmap)$, where
 $\stateSet$ is a~finite set of states,
 $\rel \subseteq \stateSet \times \stateSet$ is an edge relation,
 $s_{0} \in \stateSet$ is the initial state,
 $s_{f} \in \stateSet$ is the~terminal state, and
 $\MSCmap: \stateSet \rightarrow \mathrm{bMSC}$ is a~labeling function.
\end{definition}
  % \begin{itemize}
  % \item $\stateSet$ is a~finite set of states,
  % \item $\rel \subseteq \stateSet \times \stateSet$ is an edge relation,
  % \item $s_{0} \in \stateSet$ is an initial state,
  % \item $s_{f} \in \stateSet$ is a~terminal state, and
  % \item $\MSCmap: \stateSet \rightarrow \mathrm{MSC}$ is a~labeling
  %   function.
  % \end{itemize}

W.l.o.g., we assume that there is no incoming edge to $s_{0}$ and no outgoing edge
from $s_{f}$. Moreover, we assume that there are no nodes
unreachable
%\todo{pridal jsem, nevadi? Nebo to cele vyhodime? Je ten predpoklad
%opravdu potreba? Jojo je potreba}
from
the initial node and the~terminal node is reachable from every node in
the~graph.

%\begin{figure}[h!]
%  \begin{minipage}[b]{0.94\linewidth} 
%    \centering
%   \resizebox{10cm}{!}{\input fig/concatenate.pdf_t}
%    \caption{bMSC sequential composition}
%    \label{fig:concatenate_msc}
%  \end{minipage}
%\end{figure}

%Next we formally define runs and paths in an MSG.

Given an MSG $G=(\stateSet,\rel,s_{0}, s_{f},\MSCmap)$, a \emph{path} is
a~finite sequence of states $s_{1}s_{2} \ldots s_{k}$, where $\forall \:
1 \leq i < k : (s_{i},s_{i+1}) \in \rel$. A \emph{run} is defined as a path
with $s_{1} = s_{0}$ and $s_{k}=s_{f}$.

Intuitively, two bMSCs can be composed to a~single bMSC
by appending events of every process from the~latter bMSC at
the~end of the~process from the~precedent bMSC.
%An example can be seen in Figure~\ref{fig:concatenate_msc}, where bMSC $M_{3}$ is a~sequential composition of bMSCs $M_{1}$ and  $M_{2}$.
Formally, the \emph{sequential composition} of two bMSCs $M_{1}=(E_{1},$
$<_{1},\procSet,\type_{1},\proc_{1},\match_{1},\labels_{1})$ and
$M_{2}=(E_{2},$ $<_{2}, \procSet,\type_{2},\proc_{2},
\match_{2}, \labels_{2})$ such that the sets $E_{1}$ and $E_{2}$ are
disjoint (we can always rename events so that the sets become disjoint), is
the~bMSC $M_{1} \cdot M_{2} = ( E_{1} \cup E_{2},$ $<, \procSet,\type_{1}
\cup \type_{2}, \proc_{1} \cup \proc_{2},\match_{1} \cup
\match_{2}, \labels_{1} \cup \labels_{2})$, where $<$ is a~transitive
closure of $ <_{1} \cup <_{2} \cup \bigcup_{p \in \procSet} (\proc_{1}^{-1}
(p) \times \proc_{2}^{-1} (p))$.
Note that we consider the~weak concatenation, i.e. the~events from
the~latter bMSC may be executed even before some events from the~precedent
bMSC.  % Once we have defined paths in an MSG and how to concatenate bMSCs,

Now, we extend the~MSG labeling function $\MSCmap$ to paths. Let $\sigma = s_{1}
s_{2} \ldots s_{n}$ be a~path in MSG $G$, then
$\MSCmap(\sigma) =  \MSCmap(s_{1}) \cdot \MSCmap(s_{2}) \cdot \ldots \cdot \MSCmap(s_{n})$.
%Similarly we can extend the~linearization set. We define $\lang(G)$ for
%a~given MSG $G$:
For a given MSG $G$, the \emph{language} $\lang(G)$ is defined as $\bigcup\limits_{\sigma \:\mathrm{
    is~a~run~in } \:G} \lang(\MSCmap(\sigma))$.  Hence, two MSG are said to be
\emph{language-equivalent} if and only if they have the~same languages.
\subsubsection{Communicating Finite-State Machines}
A natural formalism for implementing bMSCs are Communicating Finite--State
Machines (CFM) that are used for example in \cite{MSC02,MSC07,BOL08}. The
CFM consists of a~finite number of finite-state machines that communicate
with each other by passing messages via unbounded FIFO channels.
\begin{definition}
  Given a finite set $\procSet$ of processes and a finite set of message
labels~$\content$, the~Communicating Finite-State Machine
  (CFM) $\automaton$ consists of Finite-State Machines (FSMs)
  $(\automaton_{p})_{p \in \procSet}.$
  Every $\automaton_p$ is a~tuple
  $(\stateSet_p,\action_p, \transition_p,s_{p}, \final_p)$, where:
  \begin{compactitem}
  \item $\stateSet_p$ is a~finite set of states,
  %\item $\action_p \subseteq \procSet \times \{!,?\} \times \procSet \times \content$ is a~finite set of actions,
  \item $\action_p \subseteq \{p!q(m) \mid q \in \procSet,m \in
\content \} \cup \{p?q(m) \mid q \in \procSet,m \in \content \}$ is a
set of actions, 
  \item $\transition_p \subseteq \stateSet_p \times \action_p \times
    \stateSet_p$ is a~transition relation,
  \item $s_{p} \in \stateSet_{p}$ is the~initial state, and
  \item $\final_p \subseteq \stateSet_p$ is a~set of local accepting states.
  \end{compactitem}
  We associate an unbounded FIFO error-free channel~$\channel_{p,q}$ with
  each pair of FSMs $\automaton_p, \automaton_q$. In every configuration,
  the~content of the~channel is a~finite word over the label alphabet
  $\content$.
\end{definition}

%\begin{figure}[ht]
%  \begin{minipage}[b]{0.94\linewidth} 
%    \centering
%   \resizebox{5cm}{!}{\input fig/cfm_example.pdf_t}
%    \caption{\emph{CFM example}}
%    \label{fig:cfm_example}
%  \end{minipage}
%\end{figure}

Whenever an FSM $\automaton_{p}$ wants to send a~message with a label $m \in
\content$ to $\automaton_{q}$, it enqueues the label $m$ into channel
$\channel_{p,q}$. We denote this action by $p!q(m)$. Provided there is
a~message with a label $m$ in the~head of channel $\channel_{p,q}$, the FSM
$\automaton_{q}$ can receive and dequeue the message with the label
$m$. This action is represented by
$q?p(m)$. % An example of a~CFM is shown in Figure~\ref{fig:cfm_example}.
A \emph{configuration} of a~CFM $\automaton = (\automaton_{p})_{p \in
  \procSet}$ is a~tuple $C = (s,B)$, where $s \in \prod_{p \in \procSet}
(\stateSet_p)$ and $B \in (\content^{*})^{\procSet \times \procSet}$ ---
local states of the FSMs together with the~contents of the~channels. Whenever
there is a configuration transition $C_i \stackrel{a_i}{\rightarrow}
C_{i+1}$, there exists a process $p \in \procSet$ such that the
FSM~$\automaton_p$ changes its local state by executing action $a_i \in A_p$
and modifies the content of one of the channels.

The CFM execution starts in an \emph{initial configuration} $s_{0} =
\prod_{p \in \procSet}\{s_{p}\}$ with all the~channels empty. The CFM is in
an \emph{accepting configuration}, if every FSM is in some of its final
states and all the~channels are empty. We will say that a~configuration is
a~\emph{deadlock}, if no accepting configuration is reachable from it. A~CFM
is \emph{deadlock-free} if no deadlock configuration is reachable from
the~initial configuration.  An \emph{accepting execution} of a CFM
$\automaton$ is a~finite sequence of configurations $C_{1}
\stackrel{a_1}{\rightarrow} C_{2} \stackrel{a_2}{\rightarrow} \ldots
\stackrel{a_{n-1}}{\rightarrow} C_{n}$ such that $C_{1}$ is the~initial
configuration and $C_{n}$ is an accepting configuration. The word $a_1a_2
\cdots a_{n-1}$ is then an \emph{accepted word} of $\automaton$.
%and for all $1 \leq i < n$, there exists a
%process $p \in \procSet$ such that the FSM $\automaton_p$ changes its local
%state by executing action $a_i \in A_p$, resulting in the configuration
%change $C_i \stackrel{a_i}{\rightarrow} C_{i+1}$.
%
%
%  We can label every execution $\sigma = C_{1}, C_{2}, \ldots,
%C_{n}$ with a~\emph{word} over alphabet $\bigcup_{p \in \procSet} \action_p$. The $i$-th
%letter of the~word is the~action labelling the~$C_{i} \rightarrow
%C_{i+1}$ transition. 
Given a CFM $\automaton$, the \emph{language} $\lang(\automaton)$ is defined
as the set of all accepted words of $\automaton$.

\section{Controllable-choice Message Sequence Graphs}
%Once the~design of the~system specified using MSG is finished, the~next step is to obtain a~distributed implementation of the~system, i.e., to realize the~system.
For a~given MSG we try to construct a~CFM such that every execution
specified in the~MSG specification can be executed by the~CFM and the~CFM
does not introduce any additional unspecified execution.

\begin{definition}[\cite{MSC07}]
  An MSG $G$ is \emph{realizable} iff there exists a~deadlock-free CFM
  $\automaton$ such that $\lang(G) = \lang(\automaton)$.
\end{definition}
%\begin{definition}[\cite{MSC07}] 
 % An MSG $G$ is \emph{safely-realizable} iff there exists a~deadlock-free CFM
  %$\automaton$, such that $\lang(G) = \lang(\automaton)$.
%\end{definition}

%In the following section we describe the~projection construction, which is a~standard realization technique.
%We show that it works surprisingly well for bMSCs.
%We formalize the~concept of adding control data into messages and show that it allows us to realize more specifications.

One of the most natural realizations are projections.  A
projection of a bMSC $M$ on a~process $p$, denoted by $M|_p$, is the sequence of events that are to
be executed by the process $p$ in $M$.  For every process $p
\in \procSet$, we construct a~FSM $\automaton_{p}$ that accepts a~single
word $M|_{p}$.  This construction is surprisingly powerful and models
all of the~bMSC linearizations. % as the~following proposition shows.
\begin{proposition}
  \label{prop:msclang} Let $M$ be a~bMSC, then CFM $\automaton = (M|_{p})_{p
    \in \procSet}$ is a~realization, i.e. $\lang(M) = \lang(\automaton)$.
\end{proposition}

It turns out that the~main obstacle when realizing MSG are nodes with
multiple outgoing edges --- choice nodes. It is necessary to ensure that all
FSMs choose the~same run through the~MSG graph. This can be hard to achieve as the~system is distributed.

%\subsection{Choice nodes in MSG}

In what follows,  we present a~known class of local-choice MSG
specifications that admits a~deadlock-free realization by adding control
data into the messages. Then, we define a~new class of
\emph{controllable-choice} MSG and compare the~expressive power of the~presented
classes.
% In this part we will present a~new class of MSG specifications --- the~controllable-choice MSGs.  The
% concepts are based on work presented in \cite{bachelor_thesis}. But,
% firstly we will introduce the simpler class of local-choice MSG.

\subsubsection{Local-choice MSG} is a class studied by many authors
\cite{MSC04,MSC35,BOL08,jard:inria-00584530,BOL01,BOL17}. Let $M$ be a~bMSC,
we say that a~process $p \in \procSet$ \emph{initiates} the bMSC $M$ if
there exists an event $e$ in $M$, such that $\proc(e) = p$ and there is no
other event $e'$ in bMSC $M$ such that $e' < e$.  For a~given MSG, every
node $s \in \stateSet$ identifies a set \emph{triggers}$(s)$, the set of
processes initiating the~communication after the node $s$. Note that it
may not be sufficient to check only the~direct successor nodes in the~MSG.

\begin{definition}
  Let $G = (\stateSet,\rel,s_{0}, s_{f},\MSCmap)$ be an MSG. For a node $s \in
  \stateSet$, the~set \emph{triggers}$(s)$ contains process $p$ if and only
  if there exists a~path $\sigma = \sigma_{1} \sigma_{2} \ldots \sigma_{n}$
  in $G$ such that $(s,\sigma_{1}) \in \rel$ and $p$ initiates bMSC
  $\MSCmap(\sigma).$
\end{definition}

\begin{definition}
  A choice node $u$ is a \emph{local-choice} node iff 
%$\vert$\emph{triggers}$(u) \vert = 1$. 
  \emph{triggers}$(u)$ is a singleton. An MSG specification $G$ is
  \emph{local-choice} iff every choice node of $G$ is local-choice.
\end{definition}

Local-choice MSG specifications have the~communication after every choice node
initiated by a~single process --- the~choice leader.
% An example of a local-choice specification can be seen in
% Figure~\ref{fig:realization_add}.
In \cite{BOL08} a deadlock-free realization with additional data in messages
is proposed.  It is easy to see that every MSG specification $G$ is
deadlock-free realizable if there is a local-choice MSG $G'$ such that
$\lang(G) = \lang(G')$. Note that the equivalence can be algorithmically
checked due to \cite{BOL17}.

\subsubsection{Controllable specifications.}
The difficulties when realizing MSG are introduced by choice nodes. In
local-choice MSG, the~additional message content is used to ensure a~single
run through the~graph is executed by all FSMs. In case of
controllable-choice MSG, the~additional content serves the~same purpose but
besides informing about the~node the~FSMs are currently executing the~FSMs
also attach a~prediction about its future execution.

This allows us to relax the~restriction on choice nodes and allows certain
non-local choice nodes to be present in the~specification.  However, it is
necessary to be able to resolve every occurrence of the~choice node, i.e.\
make the~decision in advance and inform all relevant processes.

\begin{definition}
  Let $M = (\events,\order,\procSet,\type,\proc,\match, \labels)$ be a~bMSC
  and $P' \subseteq \procSet$ be a~subset of processes. A~send event $e \in
  \events$ is a~\emph{resolving event} for $P'$ iff 
$$ \forall p \in P' \ \exists e_{p} \in \proc^{-1}(p) \textrm{ such that } e < e_{p}.$$
\end{definition}

Intuitively, resolving events of $M$ for $P'$ can distribute information to
all processes of $P'$ while executing the~rest of $M$, provided that other
processes are forwarding the~information.

\begin{definition}
  Let $G = (\stateSet,\rel,s_{0}, s_{f},\MSCmap)$ be an MSG. A~choice node
  $u$ is said to be \emph{controllable-choice} iff it satisfies both of the
  following conditions:
  \begin{itemize}
  \item For every path $\sigma$ from $s_{0}$ to $u$ there exists a~resolving
    event in bMSC $\MSCmap(\sigma)$ for \emph{triggers}$(u)$.
  \item For every path $\sigma = s_{1} s_{2} \ldots u$ such that $(u,s_{1})
    \in \rel$, there exists a~resolving event in bMSC $\MSCmap(\sigma)$ for
    \emph{triggers}$(u)$.
  \end{itemize}
\end{definition}

%\begin{wrapfigure}[17]{}{3cm}
%    \vspace{-2.5em}
  %\begin{minipage}[b]{0.94\linewidth} % a~minipage that covers half the~page
  %  \centering
%    \resizebox{3cm}{!}{\input fig/dnlc.pdf_t}
%    \caption{Controllable-choice MSG}
%    \label{fig:dnlcoutperformsnlc}
  %end{minipage}
%\end{wrapfigure}

Intuitively, a~choice node is controllable-choice, if every path from
the~initial node is labeled by a~bMSC with a~resolving event for all events
initiating the~communication after branching. Moreover, as it is necessary
to attach only bounded information, the~same restriction is required to hold
for all cycles containing a~controllable-choice node. In \cite{bachelor_thesis}
we propose an algorithm that determines whether a~given choice node is
a controllable-choice node.

\begin{definition}
  An MSG specification $G$ is \emph{controllable-choice} iff every choice
  node is either local-choice or controllable.
\end{definition}

Note that there is no bound on the~distance between the~resolving event and
the~choice node it is resolving.

\subsubsection{Local-choice vs. controllable-choice MSG.}
In the~following, we show that the controllable-choice MSG are more
expressive than local-choice MSG.  It is easy to see that every
local-choice MSG is also a~controllable-choice MSG and that not every
controllable-choice MSG is local-choice. In the following theorem, we
strengthen the result by stating that the class of MSG that are language
equivalent to some controllable-choice MSG is more expressive than the~class
of MSG that are language-equivalent to some local-choice MSG.

\begin{theorem}
  The class of MSG that are language-equivalent to some local-choice MSG,
  forms a~proper subset of MSG that are language-equivalent to some
  controllable-choice MSG.
\end{theorem}

\begin{proof}
  Consider a MSG $G = (\stateSet,\rel,s_{0}, s_{f},\MSCmap)$ with three nodes $s_0, s_f$ and $s$, such
that $(s_0,s), (s,s), (s,s_f) \in \rel$ and the only non--empty bMSC is
$\MSCmap(s)$ with two processes $p,q$. The projection
of events on $p$ is $p!q(m), p?q(m')$ and similarly for $q$ the
projection is $q!p(m'), q?p(m)$. Note that the only choice node $s$ is
controllable as both send events are resolving events for both of the
processes.

The MSG $G$ violates a necessary condition to be language equivalent to
a local-choice specification. Intuitively, the condition states that its language must be a subset
of a language of a generic local-choice equivalent MSG (for more details
see \cite{BOL17}).
\end{proof}

\section{Realizability of Controllable-choice MSG}
In this section we present an algorithm for realization of
controllable-choice MSG.  The class of local-choice specifications admits
a~natural deadlock-free realization because every branching is controlled by
a~single process.

As the~\emph{triggers} set for controllable-choice nodes can contain
multiple processes, we need to ensure that all of them reach a~consensus
about which branch to choose.  To achieve this goal, we allow the~FSMs in
certain situations to add a~behavior prediction into its outgoing
messages. Those predictions are stored in the~finite-state control units and
are forwarded within the~existing communication to other FSMs.

The length of the~prediction should be bounded, as we can attach only
bounded information to the~messages and we need to store it in
the~finite-state control unit. Therefore, it may be necessary to generate
the~behavior predictions multiple times.  As the~realization should be
deadlock-free, we must ensure that the~predictions are not conflicting ---
generated concurrently by different FSMs. To solve this we sometimes send
together with the~prediction also an event where the~next prediction should
be generated.

\begin{definition}
  A \emph{prediction} for an MSG $G = (\stateSet,\rel,s_{0}, s_{f},\MSCmap)$
  is a~pair $(\sigma,e) \in \stateSet^{*} \times (\events \cup \bot)$, where
  $\events$ is the set of all events of bMSCs assigned by $\MSCmap$, the path
  $\sigma$ is called a~\emph{prediction path}, and 
  $e$, called a \emph{control event}, is an
 % (resolving)\todo{muze tu byt
 %    resolving, nebo to radsi smazat?, pak zmenit clen na ``an''} 
  event from $\MSCmap(\sigma)$.
%
% NEBO\todo{control nebo resolving?}  
% %
%   the event $e$ is a~\emph{resolving event} of $\MSCmap(\sigma)$.
% %
  A prediction path must satisfy one of the~following conditions:
  \begin{compactitem}
  \item The prediction path $\sigma$ is the~longest common prefix of all MSG
    runs. This special initial prediction path is named $initialPath$.
  \item The prediction path $\sigma$ is the~shortest path $\sigma =
    \sigma_{1} \sigma_{2} \ldots \sigma_{n}$ in $G$ satisfying % one of
    % the~following conditions
    \begin{enumerate}
    \item $\sigma_{n} \in \local$,  or
    \item $\sigma_{n} \in \dnlc \: \wedge \: \exists \: 1 \leq i < n:
      \sigma_{i} = \sigma_{n}$, or
    \item $\sigma_{n} = s_{f}$,
    \end{enumerate}
    where $\local \subseteq \stateSet$ is the~set of all local-choice nodes
    and $\dnlc \subseteq \stateSet$ is the~set of all controllable-choice
    nodes.
  \end{compactitem}
\end{definition}

% \noindent
% To be able to access the~first and the~last MSG node of the~prediction path
% $\sigma = \sigma_{1} \sigma_{2} \ldots \sigma_{n}$ easily, we define
% functions $\mathrm{\emph{firstNode}}(\sigma) = \sigma_{1}$ and
% $\mathrm{\emph{lastNode}}(\sigma) = \sigma_{n}$.

We refer to the~first node and to the~last node of a~prediction path
$\sigma$ by $\mathrm{\emph{firstNode}}(\sigma)$ and
$\mathrm{\emph{lastNode}}(\sigma)$, respectively.

\begin{lemma}
  \label{lem:resolv}
  If the~prediction path $\sigma$ ends with a~controllable-choice node $u$,
  the~bMSC $\MSCmap(\sigma)$ contains a~resolving event for
  \emph{triggers}$(u)$ on $\MSCmap(\sigma)$.
\end{lemma}
\begin{proof}
  There are two cases to consider
  \begin{itemize}
  \item If $\sigma = \mathrm{\emph{initialPath}}$, then
    \emph{firstNode}$(\sigma) = s_{0}$ and as node $u$ is controllable-choice,
    the~path $\sigma$ contains a~resolving event for \emph{triggers}$(u)$.
  \item Otherwise, the~controllable-choice node $u$ occurs twice in the~path
    $\sigma$. As every cycle containing a~controllable-choice node has to
    contain a~resolving event for the~node, % it follows that
    there is a~resolving event for \emph{triggers}$(u)$ on path $\sigma$.
  \end{itemize}
  As there are no outgoing edges allowed in $s_{f}$, the~terminal node
  $s_{f}\not\in\dnlc$.
\qed
\end{proof}
%Next we show that for a~given MSG, the~size of a~prediction is bounded.
Note, that the number of events in a~given MSG is finite and
the~length of each~prediction path is bounded by $2 \cdot \vert \stateSet
\vert$.

When the~CFM execution starts, every FSM is initialized with an initial
prediction --- $($\emph{initialPath}$,e_{i})$ --- and starts to execute
the~appropriate projection of $\MSCmap($\emph{initialPath}$)$. The value of
$e_{i}$ depends on the~\emph{initialPath}. Let
\emph{lastNode}$($\emph{initialPath}$) = \sigma_{n}$.  In case of $\sigma_{n}
\in \dnlc$, the~event $e_{i}$ is an arbitrary resolving event from
$\MSCmap($\emph{initialPath}$)$ for \emph{triggers}$(\sigma_{n})$. It
follows from Lemma~\ref{lem:resolv} that there exists such an event.  If
$\sigma_{n} \in \local \cup \{s_{f}\}$, we set $e_{i} = \bot$.

Every FSM stores two predictions, one that is being currently executed and
a~future prediction that is to be executed after the~current one. Depending
on the~\emph{lastNode} of the~current prediction, there are the following
possibilities where to generate the~future prediction.
\begin{compactitem}
\item If \emph{lastNode} of the~current prediction is in $\local$,
  the~future prediction is generated by the~local-choice leader, while
  executing the~first event after branching.
\item If \emph{lastNode} of the~current prediction is in $\dnlc$,
  the~future prediction is generated by an FSM that executes the~control
  event of the~current prediction, while executing the~resolving event.
\item If the~\emph{lastNode} of the~current prediction is $s_{f}$, no
  further execution is possible and so no new prediction is generated.
\end{compactitem}
When an FSM generates a~new prediction, we require that there exists a
transition in the MSG from the last node of the current prediction path to
the first node of the future prediction path, as the concatenation of
prediction paths should result in a~path in the~MSG.  If an FSM
generates a~future prediction ending with a~controllable-choice node $u$, it
chooses an arbitrary resolving event for \emph{triggers}$(u)$ to be
the~resolving event in the~prediction. The existence of such an event
follows from Lemma~\ref{lem:resolv}.  To ensure that other FSMs are informed
about the~decisions, both predictions are attached to every outgoing
message. The computation ends when no FSM is allowed to generate any future
behavior.

\subsection{Algorithm}
\begin{algorithm}[p]
\caption{Process $p$ implementation}
\label{alg:impl}
\begin{algorithmic}[1]

\STATE \textbf{Variables}: $currentPrediction, nextPrediction, eventQueue;$
%\STATE
\STATE $currentPrediction \leftarrow (initialPath,e_{i});$
\STATE $nextPrediction \leftarrow \bot;$
\STATE $eventQueue \leftarrow $\textbf{push}$(\MSCmap(initialPath)|_{p});$
%\STATE
\WHILE{true}
\IF{$eventQueue$ is empty}
 \STATE \textbf{getNextNode}$();$ 
\ENDIF
\STATE $e \leftarrow \textbf{pop}(eventQueue);$

\IF{$e$ is a~send event}
 \IF{$e$ is the~resolving event in $currentPrediction$}
  %\STATE $predictionPath \leftarrow \textbf{path}(currentPrediction)$;
  \STATE $node \leftarrow \mathrm{lastNode}(\textbf{path}(currentPrediction))$;
  \STATE $nextPrediction \leftarrow \textbf{guessPrediction} (node);$
 \ENDIF
\STATE \textbf{send}$(e,currentPrediction,nextPrediction)$;
\ENDIF
\IF{$e$ is a~receive event}
 \STATE \textbf{receive}$(e,cP,nP);$
 \IF{$nextPrediction = \bot$}
 \STATE $nextPrediction \leftarrow nP;$
 \ENDIF
\ENDIF
\ENDWHILE
\end{algorithmic}
\end{algorithm}

\floatname{algorithm}{Function}

\begin{algorithm}[p]
\caption{getNextNode function for process $p$}
\label{proc:getnextnode}
\begin{algorithmic}[1]
\FUNCTION{\textbf{getNextNode}()}

\STATE  $node \leftarrow \mathrm{lastNode}(\textbf{path}(currentPrediction))$;

\IF{$node \in \dnlc \wedge p \in \mathrm{triggers}(node)$}
\STATE $currentPrediction \leftarrow nextPrediction;$
\STATE $nextPrediction \leftarrow \bot;$
\STATE $eventQueue \leftarrow $\textbf{push}$(\MSCmap(\textbf{path}(currentPrediction))|_{p});$
\ELSIF {$node \in \local \wedge p \in \mathrm{triggers}(node)$}
\STATE $currentPrediction \leftarrow \textbf{guessPrediction} (node);$
\STATE $nextPrediction \leftarrow \bot$
%\STATE $predictionPath \leftarrow \textbf{path}(currentPrediction)$;
\STATE $eventQueue \leftarrow $\textbf{push}$(\MSCmap( \textbf{path}(currentPrediction))|_{p});$
%\ELSIF {$node = s_{f}$}
%\STATE \textbf{finish();}
\ELSE
\STATE $currentPrediction \leftarrow \bot$;
\STATE $nextPrediction \leftarrow \bot$;
\STATE \textbf{polling}(); \hspace{3em}
\ENDIF

\ENDFUNCTION{}
\end{algorithmic}
\end{algorithm}

\begin{algorithm}[p]
\caption{Polling function for process $p$}
\label{proc:polling}
\begin{algorithmic}[1]
\FUNCTION{\textbf{polling}$()$}
\WHILE{true}
\IF{$p$ has a~message in some of its input buffers}
\STATE \textbf{receive}$(e,cP,nP);$
\STATE $currentPrediction \leftarrow cP;$
\STATE $nextPrediction \leftarrow nP;$
%\STATE $predictionPath \leftarrow \textbf{path}(currentPrediction)$;
\STATE $eventQueue \leftarrow $\textbf{push}$(\MSCmap(\textbf{path}(currentPrediction))|_{p});$
\STATE $\textbf{pop}(eventQueue);$
\STATE $\textbf{return;}$
\ENDIF
\ENDWHILE
\ENDFUNCTION{}
\end{algorithmic}
\end{algorithm}

In this section, we describe the~realization algorithm. All the~FSMs execute
the~same algorithm, an implementation of the~FSM $\automaton_{p}$ is
described in Algorithm~\ref{alg:impl}.
We use an auxiliary function \textbf{path} that returns a~prediction path
for a~given prediction.  Every FSM stores a~queue of events that it should
execute --- $eventQueue$. The queue is filled with projections of bMSCs
labeling projection paths --- $\MSCmap(\mathrm{prediction~path})|_{p}$ for
FSM $\automaton_{p}$. The execution starts with filling the~queue with
the~projection of the~$initialPath$.

The FSM executes a~sequence of events according to its $eventQueue$. In
order to exchange information with other FSMs, it adds its knowledge of
predictions to every outgoing message, and improves its own predictions by
receiving messages from other FSMs.

When the~FSM executes a~control event of
the~current prediction, it is responsible for generating the~next
prediction. The function \textbf{guessPrediction}$(u)$ behaves as described
in the~previous section. It chooses a~prediction $(\sigma, e)$, such that
$(u, \mathrm{firstNode}(\sigma)) \in \rel$. If \emph{lastNode}$(\sigma) \in
\dnlc$, then $e$ is a~chosen resolving event in bMSC $\MSCmap(\sigma)$ for
the~\emph{triggers} set of the~\emph{lastNode}$(\sigma)$. Otherwise, we
leave $e = \bot$.

If the~$eventQueue$ is empty, the~FSM runs the~\textbf{getNextNode} function
to determine the~continuation of the~execution.  If the~\emph{lastNode}
of the~current prediction is a~controllable-choice node and $p$ is in
the~\emph{triggers} set of this node, it uses the~prediction from its
variable $nextPrediction$ as its $currentPrediction$. The variable
$nextPrediction$ is set to $\bot$.

If the~\emph{lastNode} of the~currentPrediction is a~local-choice node and
$p$ is the~leader of the~choice, it guesses the~prediction and assigns it to
the~appropriate variables.
Otherwise, the~FSM forgets its predictions and enters a~special polling
state. This state is represented by the~\textbf{Polling} function. Whenever
the~FSM receives a~message, it sets its predictions according to
the~message. The pop function on line $8$ ensures the~consistency of
the~$eventQueue$.

An execution is finished successfully if all the FSMs are in the~polling
state and all the~buffers are empty. The correctness proof of the
following theorem is attached in the Appendix~\ref{chap:correctness}.

\begin{theorem}
  Let $G$ be a~controllable-choice MSG. Then the~CFM $\automaton$
  constructed by Algorithm~\ref{alg:impl} is a~deadlock-free realization
  i.e. $\lang(G) = \lang(\automaton)$.
\end{theorem}

\section{Conclusion}

In this work we studied the message sequence graph realizability problem,
i.e., the~possibility to make an efficient and correct distributed
implementation of the~specified system.  In general, the~problem of
determining whether a~given specification is realizable is
undecidable. Therefore, restricted classes of realizable specifications are
in a~great interest of software designers.

In recent years, a promissing research direction is to study deadlock-free
realizability allowing to attach bounded control data into existing
messages. This concept turns out to be possible to realize reasonable
specifications that are not realizable in the~very original setting.  In
this work we introduced a~new class of so called controllable-choice message
sequence graphs that admits a~deadlock-free realization with additional
control data in messages. In other words, we have sucesfully extended the
class of MSG conforming in the established setting of realizability.
Moreover, we have presented an algorithm producing realization for a given
controllable-choice message sequence graphs.
% showed that the~designed class, is the~largest known class of
% deadlock-free realizable Message Sequence Graph specifications.

\bibliographystyle{plain}                           
\bibliography{cite}{}

\newpage
\appendix
\section{Correctness}
\label{chap:correctness}

\begin{definition}[\cite{MSC07}] 
  A~word $ w \in \Sigma^{*}$ is \emph{well-formed} iff for every prefix $v$
  of $w$, every receive event in $v$ has a~matching send in $v$.  A word $ w
  \in \Sigma^{*}$ is \emph{complete} iff every send event in $w$ has
  a~matching receive event in $w$.
\end{definition}

\begin{lemma}
  \label{lem:CFMtoMSC}
  Let $\automaton$ be a~CFM and $w \in \lang(\automaton)$, then there exists
  a~bMSC $M$ such that $w \in \lang(M)$.
\end{lemma}
\begin{proof}
  Every $w \in \lang(\automaton)$ is a~well-formed and complete word. Using
  results from \cite{MSC07} a~word $w$ is a~ bMSC (potentially non-FIFO)
  linearization iff it is well-formed and complete.
  So there exists a~potentially non-FIFO bMSC $M$, such that $w \in
  \lang(M)$. It remains to show, that the~bMSC $M$ satisfies the~FIFO
  condition to fulfill our bMSC definition, but that follows directly from
  using FIFO buffers in the~CFM.
\qed
\end{proof}

% In this section we show that the~class of controllable-choice MSGs is
% deadlock-free realizable and that Algorithm~\ref{alg:impl} produces
% the~realization.
Next, we make a~few observations of the~algorithm execution.
For a~given controllable-choice MSG $G$ we construct a~CFM $\automaton =
(\automaton_{p})_{p \in \procSet}$ according to Algorithm~\ref{alg:impl}.

\begin{lemma}
  \label{lem:trigg}
  Let $(\sigma, e_{i})$ be a~prediction. FSM $\automaton_{p}$ enters
  the~\textbf{polling} function after executing $\MSCmap(\sigma)$ iff
  $$p \not \in \mathrm{\emph{triggers}} (\mathrm{\emph{lastNode}}(\sigma)).$$
\end{lemma}

\begin{proof}
  It holds for every prediction path $\sigma$ that \emph{lastNode}$(\sigma)
  \in \dnlc \cup \local \cup \{s_{f}\}$. Note that \emph{triggers}$(s_{f}) =
  \emptyset$ because no outgoing edge is allowed in the~terminal state of an
  MSG.  In case of $p \in $ \emph{triggers}, then \emph{lastNode}$(\sigma) \in
  \dnlc \cup \local$ and one of the~two branches in
  Function~\ref{proc:getnextnode} \textbf{getNextNode} is evaluated to true
  and \textbf{polling} function is skipped.
\qed
\end{proof}

It is not necessarily true that every FSM executes an event in every
prediction. In fact multiple predictions can be executed by the~CFM, while
a~particular FSM $\automaton_{p}$ executes the~polling function and is not
aware of predictions executed by other FSMs.

However, when a~prediction path ends with a~controllable-choice node, all
the~processes in the~\emph{triggers} set are active in the~prediction.

\begin{lemma}
  \label{lem:activity}
  Let $(\sigma,e_{i})$ be a~prediction, such that \emph{lastNode}$(\sigma)
  \in \dnlc$, then
$$p \in \mathrm{\emph{triggers}}(\mathrm{\emph{lastNode}}(\sigma)) \Rightarrow \MSCmap(\sigma)|_{p} \not = \emptyset$$
\end{lemma}
\begin{proof}
  Let $lastNode(\sigma)= u$. According to Lemma~\ref{lem:resolv}, there
  exists a~resolving event for \emph{triggers}$(u)$ in the bMSC
  $\MSCmap(\sigma)$. Hence, there exists an event on process $p$ that is
  dependent on the~resolving event, therefore $\MSCmap(\sigma)|_{p} \not =
  \emptyset$.
\qed
\end{proof}

Another interesting observation is that it is possible to uniquely partition
every MSG run into a~sequence of prediction paths:

\begin{proposition}
  \label{prop:partition}
  Every run $\sigma$ in $G$ can be uniquely partitioned into a~sequence of
  prediction paths such that $\sigma = initialPath \: w_{2} \ldots w_{n}$.
\end{proposition}

The following theorem shows that in fact it is not possible to execute
simultaneously different predictions by different FSMs.

\begin{theorem}
  \label{thm:agreement}
  Let $\sigma = initialPath \: w_{2} \: \ldots \: w_{n}$ such that every
  $w_{i}$ is a~prediction path. Then every FSM $\automaton_{p}$ for $p \in$
  \emph{triggers}$(lastNode(w_{n}))$ possesses the~same future prediction
  $(w_{n+1},e_{n+1})$, after executing the~last event from
  $\MSCmap(\sigma)$.
\end{theorem}
\begin{proof}
  We will prove the~theorem by induction with respect to the~length of path
  $\sigma$ (measured by the~number of prediction paths):

  \paragraph{Base case}
  Let the~length of $\sigma$ be $1$, then $\sigma = initialPath$. We have to
  consider three options, depending on the~type of the~
  $lastNode(initialPath)$:
  \begin{compactitem}
  \item Let $lastNode(initialPath) = s_{f}$, then
    \emph{triggers}$(initialPath)= \emptyset$ and there is nothing to prove.
  \item Let $lastNode(initialPath) \in \local$, then there exists a~single
    leader process in the~\emph{triggers} set. The FSM representing
    the~leader process may choose prediction $(w_{2},e_{2})$.
  \item The last option is that $lastNode(initialPath) \in \dnlc$. Then
    the~resolving event $e_{i}$ in the~initial prediction is not equal to
    $\bot$. The FSM executing the~event guesses the~next prediction
    $(w_{2},e_{2})$.

    Let $p \in \mathrm{\emph{triggers}}(lastNode(initialPath)))$. In case
    the~FSM $\automaton_{p}$ is not guessing the~prediction, we need to show
    that it receives the~prediction in some of its incoming messages. As
    $e_{i}$ is a~resolving event, there exists a dependent event on process
    $p$. Let us denote the~minimal of such events $e_{p}$. Then $e_{p}$ is
    a~receive event and it is easy to see that the~prediction
    $(w_{2},e_{2})$ is attached to the~incoming message. Hence, for every $p
    \in \mathrm{\emph{triggers}}(lastNode(initialPath)))$, FSM
    $\automaton_{p}$ has its variable $nextPrediction$ set to
    $(w_{2},e_{2})$.
\end{compactitem}
It follows from Lemma~\ref{lem:trigg} that for every $p$ not in
the~\emph{triggers} set, the~FSM $\automaton_{p}$ is in the~polling state
having its variable $nextPrediction$ set to $\bot$.

\paragraph{Induction step}
Let the~length of $\sigma$ be $n$. As in the~base case, we have to consider
multiple options:

\begin{itemize}
\item Let $lastNode(w_{n}) \in \{s_{f}\} \cup \local$, then the~argument is
  the~same as in the~base case.
\item So let $lastNode(w_{n}) \in \dnlc$. From induction hypothesis, it follows
  that all FSMs $\automaton_{p}$ for $p \in
  \mathrm{\emph{triggers}}(w_{n-1})$, start to execute prediction path
  $w_{n}$ and all the~others are in the~polling state.

  Let $p \in \mathrm{\emph{triggers}}(w_{n})$. We show that FSM
  $\automaton_{p}$ executes the~projection $\MSCmap(w_{n})|_{p}$. It follows
  from Lemma~\ref{lem:activity} that this projection is non-empty. We have
  already shown that this is true for FSMs $\automaton_{p}$, such that $p
  \in \mathrm{\emph{triggers}}(w_{n-1})$. In case of $p \not \in
  \mathrm{\emph{triggers}}(w_{n-1})$, the~FSM $\automaton_{p}$ is in
  the~polling state. As it is not in the~\emph{triggers} set, its first
  action is a~receive event. It is easy to see that the~incoming message
  already contains the~current prediction $(w_{n},e_{n})$ and FSM
  $\automaton_{p}$ starts to execute $\MSCmap(w_{n})|_{p}$.

  The rest of the~proof is similar to the~base case. During the~execution of
  the~resolving event $e_{n}$ a~new prediction $(w_{n+1},e_{n+1})$ is guessed
  and distributed to all FSMs $\automaton_{p}$ for $p \in
  \mathrm{\emph{triggers}}(w_{n})$.

\end{itemize}
\qed
\end{proof}

To show that Algorithm~\ref{alg:impl} is a~deadlock-free realization of
the~class of control\-lable-choice MSG we need to show that $\lang(G) =
\lang(\automaton)$. We will divide the~proof into two parts, first showing
that $\lang(G) \subseteq \lang(\automaton)$ and finishing with
$\lang(\automaton) \subseteq \lang(G)$.

\subsection{$\lang(G) \subseteq \lang(\automaton)$}

We show that for all $w \in \lang(G)$ also holds that $w \in
\lang(\automaton)$. For every $w \in \lang(G)$ there exists a~run $\sigma$
in $G$ such that $w \in \lang(\MSCmap(\sigma))$.

We need to find a~CFM execution, such that every FSM $\automaton_{p}$
executes the~projection $\MSCmap(\sigma)|_{p}$ and ends in a~polling state
with the~CFM having all the~channels empty. Then using
Proposition~\ref{prop:msclang}, follows $\lang(M) \subseteq
\lang(\automaton)$ and especially $w \in \lang(\automaton)$.

According to Proposition~\ref{prop:partition} we can partition every run
$\sigma$ uniquely into a~sequence of prediction paths --- $initialPath \:
w_{2} \ldots w_{n}$. This sequence is a~natural candidate for prediction
paths that should be guessed during the~CFM execution.

Every CFM execution starts with an initial prediction
$(initialPath,e_{i})$. The guessed future prediction paths are $w_{2}, w_{3}
\ldots$. The guessing continues until the~last prediction path $w_{n}$ is
executed. As $\sigma$ is a~run in MSG $G$, $lastNode(w_{n}) =
s_{f}$. Therefore, \emph{triggers}$(lastNode(w_{n})) = \emptyset$. It
follows from Lemma~\ref{lem:trigg} that all the~FSMs are in the~polling
state. All the~channels are empty because of the~well-formedness and
the~completeness of the~bMSC linearizations.

\subsection{$\lang(\automaton) \subseteq \lang (G)$}
We show that for every $w \in \lang(\automaton)$ also $w \in
\lang(G)$. According to Lemma~\ref{lem:CFMtoMSC}, every $w \in
\lang(\automaton)$ identifies a~bMSC $M$. To conclude this part of
the~proof, we find a~run $\sigma$ in $G$, such that $M =
\MSCmap(\sigma)$. As $\lang(M) \subseteq \lang(G)$ we get $w \in \lang(G)$.

The $\sigma$ run in $G$ is defined inductively. Every FSM starts with
executing the~$initialPath$ prediction path. So it is safe to start the~run
$\sigma$ with this prediction path.

According to Theorem~\ref{thm:agreement} whenever some prediction $w_{i}$ is
executed, all FSMs $\automaton_{p}$ for $p \in
\mathrm{\emph{triggers}}(\mathrm{lastNode}(w_{i}))$ agree on some future
prediction $w_{i+1}$ and all $\automaton_{p}$ such that $p$ executes an
event in bMSC $\MSCmap(w_{i+1})$, execute the~projections
$\MSCmap(w_{i+1})|_{p}$. All the~other FSMs are in the~polling state and are
awakened only if needed.

The predictions are guessed in such a way that the~following condition holds:

$$(\mathrm{lastNode}(w_{i}),\mathrm{firstNode}(w_{i+1})) \in \rel$$

So it is safe to append $w_{i+1}$ at the~end of $\sigma$.  Next we show that
$\sigma$ ends with a~terminal node. The CFM accepts when all the~channels
are empty and all the~FSMs are in the~polling state. Hence, the~last
prediction that was executed ended with a~node with an empty triggers
set. In general it is possible that this is may not be the~terminal node,
but every path from this node reaches $s_{f}$ without executing any
event. So we can safely extend $\sigma$ with a~path to a~terminal node.

\end{document}